\documentclass{article}

\usepackage[english]{babel}
\usepackage[utf8]{inputenc}
\usepackage[T1]{fontenc}
\usepackage{calligra}
\usepackage{amsfonts,amsmath,amssymb}
\usepackage{amsthm}
\usepackage{mathabx}
\usepackage[usenames]{color} 
\usepackage{xspace}
\usepackage[svgnames]{xcolor}
\usepackage{subfig,caption}
\usepackage{todonotes}
\usepackage{paralist}
\usepackage{microtype}
\usepackage{tikz}
\usetikzlibrary{shapes}
\usepackage{bookman}
\usepackage{authblk}

\newif\ifdraft\drafttrue

\ifdraft

\newcommand\os[1]{\todo[inline,size=\scriptsize,backgroundcolor=PaleTurquoise]{#1 - \textbf{Olivier}}}
\newcommand\dc[1]{\todo[inline,size=\scriptsize,backgroundcolor=Yellow]{#1 - \textbf{Didier}}}
\newcommand\ac[1]{\todo[inline,size=\scriptsize,backgroundcolor=SpringGreen]{#1 - \textbf{Arnaud}}}

\newcommand\vlong[1]{\todo[inline,size=\scriptsize,backgroundcolor=red, caption={2do}]{
\begin{minipage}{\textwidth-4pt}#1  - \textbf{Garder dans la version longue}\end{minipage}}}
\else

\newcommand\os[1]{}
\newcommand\ac[1]{}
\newcommand\dc[1]{}
\newcommand\review[1]{}

\newcommand\vlong[1]{}
\fi

\newcommand{\ie}{\emph{i.e.}\xspace}

\renewcommand{\epsilon}{\varepsilon}
\renewcommand{\phi}{\varphi}

\newtheorem{theorem}{Theorem}
\newtheorem{example}{Example}

\newtheorem{corollary}{Corollary}

\newcommand\textbfit[1]{\textbf{\em #1}}
\newcommand{\defin}[1]{\textbfit{\boldmath #1}}
\newcommand{\eg}{\emph{e.g.}\xspace}
\newcommand{\resp}{resp.\xspace} 

\newcommand{\fl}[1]{\overset{#1}{\longrightarrow}}
\newcommand{\qini}{q_{in}}
\newcommand{\Qfinal}{Q_{fin}}
\newcommand{\qfinal}{q_{fin}}
\newcommand{\qs}{q_{\sharp}}
\newcommand{\wG}{\widetilde{G}}
\begin{document}


\title{Marking Shortest Paths On Pushdown Graphs Does Not Preserve MSO Decidability}

\author[1]{Arnaud Carayol\thanks{{Arnaud.Carayol@univ-mlv.fr}}}
\author[2]{Olivier Serre\thanks{{Olivier.Serre@cnrs.fr}}}
\affil[1]{LIGM (CNRS \& Université Paris Est)}
\affil[2]{IRIF (CNRS \& Université Paris Diderot -- Paris 7)}
\date{}


\maketitle

\begin{abstract}
In this paper we consider pushdown graphs, \ie infinite graphs that can be described as transition graphs of deterministic real-time pushdown automata. We consider the case where some vertices are designated as being final and we build, in a breadth-first manner, a marking of edges that lead to such vertices ({\ie, for every vertex that can reach a final one, we mark all out-going edges laying on some shortest path to a final vertex}).

Our main result is that the edge-marked version of a pushdown graph may itself no longer be a pushdown graph, 
{as we prove that the MSO theory of this enriched graph may be undecidable}.
\end{abstract}
\newpage

\section{Introduction}

The original motivation of this paper comes from the following work plan: design algorithms working on \emph{infinite} graphs and computing classical objects from graph theory. 
One {firstly} targeted set of algorithms are naturally those computing spanning trees, \emph{e.g.} spanning trees built by performing a breadth-first search or the ones built by performing a depth-first search. In particular, breadth-first search seems to be a good candidate as it can  easily be defined by a smallest fixpoint computation. 

Of course, to ensure termination one has to identify reasonable classes of infinite graphs: obviously such a class should provide finite description of its elements and the graphs should have some good decidability properties. {The simplest such class is the class of} transition graphs of pushdown automata: they are finitely described by the underlying pushdown automata and they enjoy many good properties, in particular with respect to logic and games (see \eg~\cite{MullerS85,Walukiewicz01}). {In particular, monadic second-order logic (MSO) is decidable for any pushdown graph.}

Another expected property of our algorithm is that it should be reflective\footnote{In programming languages, reflection is the process by which a computer program can observe and dynamically modify its own structure and behaviour. See~\cite{BCOS10} for an example of reflection in the richer setting of collapsible pushdown automata and recursion schemes.} in the following sense: the produced outputs should belong to the same class of structures {as} the inputs. Equivalently, in the setting of pushdown graphs, it means that we want to design an algorithm that takes as an input a pushdown graph and produces as an output another pushdown graph that is {an isomorphic copy of} the input graph enriched with a marking of some edges that corresponds to a breadth-first search spanning tree.

The main result of this paper is that such an algorithm does not exist, \ie there is no algorithm that takes as an input a pushdown graphs and returns a copy of it marked with a breadth-first search spanning tree. The roadmap to prove this result is to exhibit a pushdown graph such that when marked with a breadth-first search spanning tree leads to a graph with an undecidable MSO theory: as pushdown graphs enjoy decidable MSO theories it directly permits to conclude.

The paper starts by introducing in Section~\ref{sec:def} the classical objects and formally defines the problem under study. Our main results are proven in Section~\ref{sec:main} while we briefly discuss some consequences in Section~\ref{ref:consequences}.


\section{Preliminaries}\label{sec:def}

An \defin{alphabet} $A$ is a finite set of letters. In the sequel $A^*$ denotes the set of finite words over $A$
and the \defin{empty word} is written $\epsilon$. The length of a word $u$ is denoted by $|u|$ and for any $k\geq 0$, we let 
$A^{\leq k}=\{u\mid |u|\leq k\}$.
 Let $u$ and $v$ be two finite words. Then $u\cdot v$ (or simply $uv$) denotes the \defin{concatenation} of $u$ and $v$.

Let $A$ be an alphabet. An $A$-labeled (oriented) \defin{graph} is a {pair} $G=(V,E)$ where $V$ is a (possibly infinite) set of vertices and $E\subseteq V\times A\times V$ is a (possibly infinite) set of edges. In the sequel we write $v\fl{a} v'$ to denote that $(v,a,v')\in E$.

{A vertex $v'$ is \defin{reachable} from a vertex $v$ if there is a sequence $v_1,\dots,v_\ell$ of vertices together with a sequence of letters $a_1,\dots,a_{\ell-1}$ such that $v_1=v$, $v_\ell=v'$ and $v_i\fl{a_i}v_{i+1}$ for every $i=1,\dots,\ell-1$.}

\subsection{Pushdown Graphs}\label{ssec:pushdown}

A  \defin{deterministic real-time pushdown automaton} is defined as a tuple $\mathcal{P}=(Q,A,\Gamma,\bot,\qini,\Qfinal,\delta)$ where $Q$ is a finite set of control states, $A$ is a finite input alphabet, $\Gamma$ is a finite stack alphabet, $\bot\in \Gamma$ is a bottom-of-stack symbol, $\qini\in Q$ is an initial state, $\Qfinal\subseteq Q$ is a set of final states and $\delta: Q\times \Gamma\times A\rightarrow Q\times \Gamma^{\leq 2}$ is a \emph{partial} transition function such that
\begin{itemize}
\item $\delta(q,\bot,a)=(q',u) \Rightarrow u\in (\Gamma\setminus\{\bot\})\bot\cup \{\bot\}$, \emph{i.e.} $\bot$ cannot be removed.
\item $\delta(q,\gamma,a) = (q',u)$ and $\gamma\neq \bot \Rightarrow u\in (\Gamma\setminus\{\bot\})^{\leq 2}$, \emph{i.e.}  $\bot$ cannot be pushed. 
\end{itemize}

A \defin{configuration} of $\mathcal{P}$ is a pair $(q,\sigma)\in Q\times (\Gamma\setminus\{\bot\})^*\bot$ consisting of a control state and a well-formed stack content. The \defin{initial configuration} of $\mathcal{P}$ is $(\qini,\bot)$ and the \defin{final configurations} of $\mathcal{P}$ are those of the form $(\qfinal,\bot)$ with $\qfinal\in \Qfinal$.

Let $(q,\sigma)$ and $(q',\sigma')$ be two configurations, and let $a\in A$ be a letter. Then, there is an \defin{$a$-labelled transition} from $(q,\sigma)$ to $(q',\sigma')$, denoted $(q,\sigma)\fl{a}(q',\sigma')$, if and only if one has $\delta(q,\gamma,a)=(q',u)$ where $\sigma = \gamma\sigma''$ and $\sigma'=u\sigma''$, \ie $\sigma'$ is obtained from $\sigma$ by replacing its top symbol $\gamma$ by $u$.

The \defin{configuration graph} of $\mathcal{P}$ is the $A$-labeled graph $G_{\mathcal{P}} = (V_{\mathcal{P}},E_{\mathcal{P}})$ where $V_{\mathcal{P}}$ is the set of configurations of $\mathcal{P}$ and where $E_{\mathcal{P}}$ is the transition {relation} defined by $\mathcal{P}$. {A graph isomorphic to a graph $G_{\mathcal{P}}$ is called a \defin{pushdown graph}.}

\begin{example}\label{ex:pda}
As a running example, consider the following pushdown automaton $\mathcal{P}=(Q,A,\Gamma,\bot,\qini,\{\qfinal\},\delta)$ where one lets $Q=\{\qini,\qfinal,\qs\}$, $A=\{a,b,\sharp\}$, $\Gamma=\{a,b,\bot\}$ and $\delta$ be as follows:
\begin{itemize}
\item $\delta(\qini,\gamma,a) = (\qini,a\gamma)$ and $\delta(\qini,\gamma,b) = (\qini,b\gamma)$: in the initial state on reading a symbol in $\{a,b\}$ it is copied on top of the stack.
\item $\delta(\qini,\gamma,\sharp) = (\qs,\gamma)$: in the initial state on reading symbol $\sharp$ the state is switched to $\qs$.
\item For $x\in\{a,b\}$ and $\gamma\neq \bot$, $\delta(\qs,\gamma,x) = (\qs,\epsilon)$ if $\gamma=x$ and $\delta(\qs,\gamma,x) = (\qs,x\gamma)$ if $\gamma\neq x$; and $\delta(\qs,\gamma,\sharp) = (\qs,\gamma)$: in the state $\qs$ an input letter $\sharp$ does not change the configuration while for an input letter in $\{a,b\}$ the top symbol is popped if it is the same as the input symbol otherwise the letter is copied on top of the stack.
\item $\delta(\qs,\bot,x) = (\qfinal,{\bot})$ for any $x\in A$: once the stack is emptied in the state $\qs$ one goes to the state $\qfinal$.
\item $\delta(\qfinal,\bot,x) = (\qfinal,{\bot})$ for any $x\in A$: once the configuration $(\qfinal,\bot)$ is reached it stays in forever.
\end{itemize}
The graph $G_{\mathcal{P}}$ is depicted in Figure~\ref{fig:example}.
\end{example}

\begin{figure}
\centering
\begin{tikzpicture}[scale=.9,transform shape]
\node (qsbot) at (1,1.5) {\small$(\qs,\bot)$};
\node (qfbot) at (1,3) {\small$(\qfinal,\bot)$};
\node (qibot) at (1,0) {\small$(\qini,\bot)$};
\node (qia) at (-1.5,-1.5) {\small$(\qini,a\bot)$};
\node (qib) at (3.5,-1.5) {\small$(\qini,b\bot)$};
\path[->] (qibot) edge node[above left] {\small $a$} (qia);
\path[->] (qibot) edge node[above right] {\small $b$} (qib);
\path[->] (qibot) edge node[right] {\small $\sharp$} (qsbot);
\node (qiaa) at (-2.5,-3) {\small$(\qini,aa\bot)$};
\node (qiba) at (-.5,-3) {\small$(\qini,ba\bot)$};
\node (qiab) at (2.5,-3) {\small$(\qini,ab\bot)$};
\node (qibb) at (4.5,-3) {\small$(\qini,bb\bot)$};
\node (qsa) at (-4,-1.5) {\small$(\qs,a\bot)$};
\node (qsba) at (-5,-3) {\small$(\qs,ba\bot)$};
\node (qsb) at (6,-1.5) {\small$(\qs,b\bot)$};
\node (qsab) at (7,-3) {\small$(\qs,ab\bot)$};
\node at (-2.5,-3.3) {\small \vdots};
\node at (-.5,-3.3) {\small \vdots};
\node at (-5,-3.3) {\small \vdots};
\node at (2.5,-3.3) {\small \vdots};
\node at (4.5,-3.3) {\small \vdots};
\node at (7,-3.3) {\small \vdots};
\path[->] (qia) edge node[left] {\small $a$} (qiaa);
\path[->] (qia) edge node[right] {\small $b$} (qiba);
\path[->] (qia) edge node[above] {\small $\sharp$} (qsa);
\path[->] (qib) edge node[above] {\small $\sharp$} (qsb);
\path[->] (qsa) edge[bend left] node[above left] {\small $a$} (qsbot);
\path[->] (qsb) edge[bend right] node[above right] {\small $b$} (qsbot);
\path[->] (qib) edge node[left] {\small $a$} (qiab);
\path[->] (qib) edge node[right] {\small $b$} (qibb);
\path[->] (qsbot) edge node[right] {\small $a,b,\sharp$} (qfbot);
\path[->] (qfbot) edge [loop right] node {\small $a,b,\sharp$} (qfbot);
\path[->] (qsa) edge [loop left] node {\small $\sharp$} (qsa);
\path[->] (qsb) edge [loop right] node {\small $\sharp$} (qsb);
\path[->] (qsa) edge node[left] {\small $b$} (qsba);
\path[->] (qsba) edge[bend right] node[right] {\small $b$} (qsa);
\path[->] (qiba) edge[bend left] node[below] {\small $\sharp$} (qsba);
\path[->] (qsab) edge[bend left] node[left] {\small $a$} (qsb);
\path[->] (qsb) edge node[right] {\small $a$} (qsab);
\path[->] (qiab) edge[bend right] node[below] {\small $\sharp$} (qsab);
\end{tikzpicture}
\caption{The configuration graph $G_{\mathcal{P}}$ of the pushdown automaton of Example~\ref{ex:pda}}\label{fig:example}
\end{figure}

We are interested in defining, in a breadth-first search manner, the set of configurations from which one can reach a final configuration. For this consider the following increasing sequence $(W_i)_{i\geq 0}$ of configurations of $\mathcal{P}$ and call its limit $W$.
\begin{itemize}
\item $W_0 = \{(\qfinal,\bot)\mid \qfinal\in\Qfinal\}$ consists only of the final configurations.
\item $W_{i+1} = W_i\cup \{(q,\sigma)\mid \exists (q',\sigma')\in W_i\text{ and }a\in A \text{ s.t. }(q,\sigma)\fl{a}(q',\sigma')\}$.
\end{itemize}
Obviously, $W$ is the set of all configurations from which a final configuration is reachable. Define for every configuration $(q,\sigma)$ its rank $rk((q,\sigma))$ to be the smallest $i$ such that $(q,\sigma)\in W_i$ when exists and to be $\infty$ otherwise.

We now define a new graph $\wG_{\mathcal{P}}$ obtained from $G_{\mathcal{P}}$ by marking those edges that go from a configuration to one with a strictly smaller rank (equivalently that decrease the rank by $1$). First we let $\widetilde{A}=A\cup\{\underline{a}\mid a \in A\}$ consists of $A$ together with a marked copy of each of its elements. Then we let $\wG_{\mathcal{P}}$ be the $\widetilde{A}$-labelled graph $(V_\mathcal{P},\widetilde{E}_\mathcal{P})$ where
\begin{itemize}
\item $((q,\sigma),a,(q',\sigma'))\in \widetilde{E}_\mathcal{P}$ if $((q,\sigma),a,(q',\sigma'))\in E_\mathcal{P}$ and $rk((q',\sigma'))\geq rk((q,\sigma))$;
\item $((q,\sigma),\underline{a},(q',\sigma'))\in \widetilde{E}_\mathcal{P}$ if $((q,\sigma),a,(q',\sigma'))\in E_\mathcal{P}$ and $rk((q',\sigma'))< rk((q,\sigma))$.
\end{itemize}

Coming back to Example~\ref{ex:pda}, the graph $\wG_{\mathcal{P}}$ is depicted in Figure~\ref{fig:example-continue}.

\begin{figure}
\centering
\begin{tikzpicture}[select/.style={very thick},unselect/.style={dotted},rouge/.style={red},scale=.9,transform shape]
\node (qsbot) at (1,1.5) {\small$(\qs,\bot)$};
\node (qfbot) at (1,3) {\small$(\qfinal,\bot)$};
\node (qibot) at (1,0) {\small$(\qini,\bot)$};
\node (qia) at (-1.5,-1.5) {\small$(\qini,a\bot)$};
\node (qib) at (3.5,-1.5) {\small$(\qini,b\bot)$};
\path[->,unselect] (qibot) edge node[above left] {\small $a$} (qia);
\path[->,unselect] (qibot) edge node[above right] {\small $b$} (qib);
\path[->,select] (qibot) edge node[right] {\small $\sharp$} (qsbot);
\node (qiaa) at (-2.5,-3) {\small$(\qini,aa\bot)$};
\node (qiba) at (-.5,-3) {\small$(\qini,ba\bot)$};
\node (qiab) at (2.5,-3) {\small$(\qini,ab\bot)$};
\node (qibb) at (4.5,-3) {\small$(\qini,bb\bot)$};
\node (qsa) at (-4,-1.5) {\small$(\qs,a\bot)$};
\node (qsba) at (-5,-3) {\small$(\qs,ba\bot)$};
\node (qsb) at (6,-1.5) {\small$(\qs,b\bot)$};
\node (qsab) at (7,-3) {\small$(\qs,ab\bot)$};
\node at (-2.5,-3.3) {\small \vdots};
\node at (-.5,-3.3) {\small \vdots};
\node at (-5,-3.3) {\small \vdots};
\node at (2.5,-3.3) {\small \vdots};
\node at (4.5,-3.3) {\small \vdots};
\node at (7,-3.3) {\small \vdots};
\path[->,unselect] (qia) edge node[left] {\small $a$} (qiaa);
\path[->,unselect] (qia) edge node[right] {\small $b$} (qiba);
\path[->,select] (qia) edge node[above] {\small $\underline{\sharp}$} (qsa);
\path[->,select] (qib) edge node[above] {\small $\underline{\sharp}$} (qsb);
\path[->,select] (qsa) edge[bend left] node[above left] {\small $\underline{a}$} (qsbot);
\path[->,select] (qsb) edge[bend right] node[above right] {\small $\underline{b}$} (qsbot);
\path[->,unselect] (qib) edge node[left] {\small $a$} (qiab);
\path[->,unselect] (qib) edge node[right] {\small $b$} (qibb);
\path[->,select] (qsbot) edge node[right] {\small $\underline{a},\underline{b},\underline{\sharp}$} (qfbot);
\path[->,unselect] (qfbot) edge [loop right] node {\small $a,b,\sharp$} (qfbot);
\path[->,unselect] (qsa) edge [loop left] node {\small $\sharp$} (qsa);
\path[->,unselect] (qsb) edge [loop right] node {\small $\sharp$} (qsb);
\path[->,unselect] (qsa) edge node[left] {\small $b$} (qsba);
\path[->,select] (qsba) edge[bend right] node[right] {\small $\underline{b}$} (qsa);
\path[->,select] (qiba) edge[bend left] node[below] {\small $\underline{\sharp}$} (qsba);
\path[->,select] (qsab) edge[bend left] node[left] {\small $\underline{a}$} (qsb);
\path[->,unselect] (qsb) edge node[right] {\small $a$} (qsab);
\path[->,select] (qiab) edge[bend right] node[below] {\small $\underline{\sharp}$} (qsab);

\end{tikzpicture}
\caption{The graph $\wG_{\mathcal{P}}$ of the pushdown automaton of Example~\ref{ex:pda}}\label{fig:example-continue}
\end{figure}
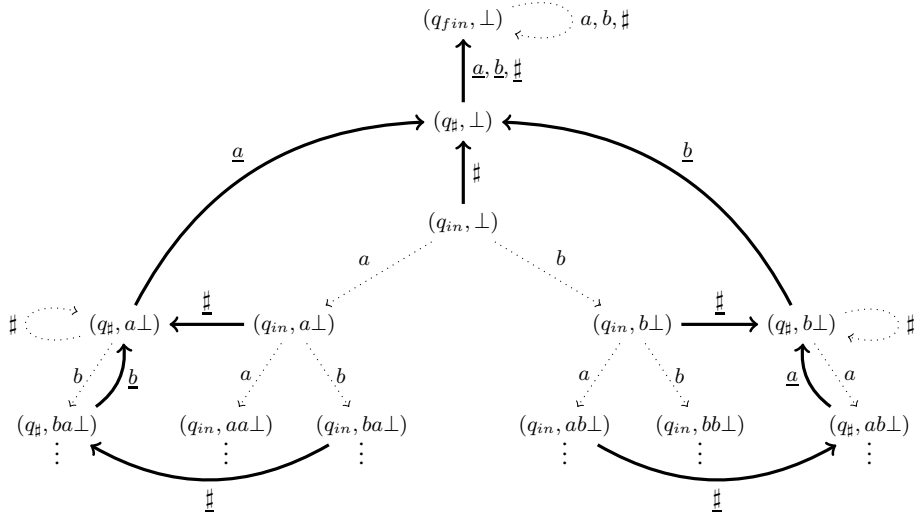

Finally, one can consider a graph built out of $G_{\mathcal{P}}$ by marking only some (but at least one) shortest paths to a final configuration (the extreme case being when the marked paths form a spanning tree). More precisely, a \defin{well-formed marking} of $G_{\mathcal{P}}$ is an $\widetilde{A}$-labelled graph $G=(V_{\mathcal{P}},E)$ such that:
\begin{itemize}
{\item For every edge $((q,\sigma),\underline{a},(q',\sigma'))\in E_{\mathcal{P}}$, one has either $((q,\sigma),a,(q',\sigma'))\in E$ or $((q,\sigma),\underline{a},(q',\sigma'))\in E$.}
{\item For every edge $((q,\sigma),a,(q',\sigma'))\in E_{\mathcal{P}}$, one has  $((q,\sigma),a,(q',\sigma'))\in E$.}
{\item For every edge $((q,\sigma),a,(q',\sigma'))\in E$, one has either $((q,\sigma),a,(q',\sigma'))\in E_{\mathcal{P}}$ or $((q,\sigma),\underline{a},(q',\sigma'))\in E_{\mathcal{P}}$.
}
{\item For every edge $((q,\sigma),\underline{a},(q',\sigma'))\in E$, one has $((q,\sigma),\underline{a},(q',\sigma'))\in E_{\mathcal{P}}$.
}
\item For every configuration $(q,\sigma)\in W$ one has at least one edge of the form $((q,\sigma),\underline{a},(q',\sigma'))\in E$.
\end{itemize}

\subsection{Monadic Second Order Logic}

\defin{Monadic Second Order Logic (MSO)} is a classical logical formalism {(extending First Order Logic)} to express properties of a relational structure. In this framework, formulas are built from atomic formulas using Boolean connectives (negation, disjunction, conjunction) and quantifiers. Atomic formulas in first-order logic have either the form $x_1=x_2$ or $x_1\fl{a}x_2$ where $x_1$ and $x_2$ are first-order variables (that each stands for an element in the domain, \ie for a vertex); in monadic second-order logic there are additional atomic formulas of the form $x\in X$ where $x$ is a first-order variable and $X$ is a second-order variable (that stands for a subset of the domain, \ie for a subset of vertices). First-order (\resp second-order) variables are introduced thanks to existential and universal quantifications. 

Whether a formula (without free variable) holds in a structure is defined as usual, and to keep this article short we refer the reader \eg to \cite{Thomas97} for examples and formal definitions. When a formula has free variables it permits to express whether a structure together with an assignation of the free variables satisfy the formula: in particular if one has a single {first-order} free variable, it permits to express a property of an element inside a structure. 

We say that a structure has a \defin{decidable MSO theory} in case the following problem is decidable: does the input formula $\phi$ hold in the structure? 

The following is a classical result due to Muller and Schupp \cite{MullerS85}

\begin{theorem}\label{theo:MSOpushdown}
For any pushdown automaton $\mathcal{P}$ the graph $G_{\mathcal{P}}$ has a decidable MSO theory.
\end{theorem}

\section{Main Results}\label{sec:main}

We are now ready to state our first result.

\begin{theorem}\label{theo:full}
There is a pushdown automaton $\mathcal{P}$ such that the graph $\wG_{\mathcal{P}}$ has an undecidable MSO theory.
\end{theorem}

\begin{proof}
The idea is to design a pushdown automaton $\mathcal{P}$ such that, for any 2-counter machine with tests for $0$, one can build an MSO formula that is true in $\wG_{\mathcal{P}}$ if and only if the 2-counter machine halts from its initial configuration. As the halting problem is undecidable for 2-counter machine with tests for $0$, it follows that $\wG_{\mathcal{P}}$ has an undecidable MSO theory.

We define a pushdown automaton over the stack alphabet $\Gamma=(\{1,2,3\}\times\{1,2,3\}\times\{2\})\cup\{\bot\}$: with any element over $\Gamma^*\bot$ one can associate a triple $(k_1,k_2,k_3)$ where $k_i$ is obtained by taking the sum along the $i$-th component of the elements in the stack (ignoring $\bot$). With such a triple $(k_1,k_2,k_3)$ we can associate a pair of counters $(k_1-k_3,k_2-k_3)\in\mathbb{Z}\times\mathbb{Z}$. Of course several stack contents may encode the same counters values but this is not problematic later on.

The pushdown automaton $\mathcal{P}$ works in two modes (the mode being indicated by the control state; we omit the states in this description to help readability). In the first mode, one can push any symbol {in $\{1,2,3\}\times\{1,2,3\}\times\{2\}$} (namely there is a dedicated input letter for any symbol to be pushed), which allows to increment/decrement/keep unchanged the two counters: \emph{e.g.} to increment counter $k_1-k_3$ and decrement the counter $k_2-k_3$ one pushes $(3,1,2)$; more generally to modify the counter $k_1-k_3$ by $\iota_1$ and to modify the counter $k_2-k_3$ by $\iota_2$ one pushes $(2+\iota_1,2+\iota_2,2)$. In the first mode, one can switch to a second mode (thanks to a special input letter) that comes with three variants depending on the control state (call those states $p_1,p_2$ or $p_3$). In the state $p_i$ one pops on the stack and the “popping speed” depends on the $i$-th component of the current top symbol: if the top stack symbol {has $1$ as its $i$-th component, one simply pops, if it has $x>1$ one goes first to $x-1$ intermediate states before popping}. Hence, in a configuration with associated triple $(k_1,k_2,k_3)$ it takes $k_i$ steps before emptying the stack starting in the state $p_i$. Once the stack is emptied, the (unique) final state is reached (that is from a configuration with top symbol $\bot$ there is only one possible transition that goes to the final state and this is the only way to reach it).

{Now, consider the graph $\wG_{\mathcal{P}}$ and let us explain how the extra information carried by this graph permits to compare $k_1,k_2$ and $k_3$. This in turn will allow us to test if  the counters $k_1-k_3$ and $k_2-k_3$ are equal to $0$ when simulating a 2-counter machine.
In the following, we refer to an edge in $\wG_{\mathcal{P}}$ as being marked if it corresponds to an underlined input symbol.}

Consider a configuration (in the first mode) with an associated triple $(k_1,k_2,k_3)$ and assume that one wants to check whether $k_1=k_3$. For this one looks at the marked outgoing edges in the present configuration: if there is one that goes to a configuration (in the second mode) with the state $p_1$ and another one that goes to a configuration with the state $p_3$ (remember that the marked edges indicate the fastest choices to the final configuration) one directly concludes that $k_1=k_3$, hence that $k_1-k_3=0$; however one may have $k_1=k_3$ without being in the previous situation namely if $k_2<k_1,k_3$. In this latter case the trick is to increase all three components: the $k_1$ and the $k_3$ components by the same value and the $k_2$ component by a much bigger value, leading to a configuration $(k_1',k_2',k_3')$ with $k_2'>k_1',k_3'$ and $k_1'=k_3'$ if and only if $k_1=k_3$. The latter can easily be achieved by performing a sequence of actions pushing $(2,3,2)$. Hence, if one wants to check whether $k_1=k_3$ it suffices to check whether the following holds: “either the marked edges indicate to go both to $p_1$ and $p_3$ or there exists a path along which only pushing $(2,3,2)$ are performed and that leads to a configuration where the marked edges indicate to go both to $p_1$ and $p_3$”. As the latter can easily be stated in MSO logic it follows that one can write an MSO formula (with one order-1 free variable) that holds exactly in those configurations (in the first mode) where $k_1=k_3$. Similarly, one can design a formula to check whether $k_2=k_3$.

Now for any 2-counter machine with tests for $0$ one can easily write an MSO formula that builds on the previous formulas and checks whether there is a run of the machine that is halting: the states of the machine are directly handled in the MSO formula while the counter is managed thanks to the underlying structure of $\wG_{\mathcal{P}}$.

As the halting problem for 2-counter machine with tests for $0$ is undecidable, it concludes the proof.
\end{proof}

We now state our second result which is an analog of Theorem~\ref{theo:full} but for well-formed marking of $G_{\mathcal{P}}$ (hence, it encompasses Theorem~\ref{theo:full} as $\wG_{\mathcal{P}}$ is a well-formed marking of $G_{\mathcal{P}}$).

\begin{theorem}\label{theo:sub}
There is a pushdown automaton $\mathcal{P}$ such that any well-formed marking of $G_{\mathcal{P}}$ has an undecidable MSO theory.
\end{theorem}

\begin{proof}
The proof used to establish Theorem~\ref{theo:full} no longer works because \emph{e.g.} we could have $k_1=k_3$ but only the edge leading to the configuration with the state $p_1$ is marked. However, the same pushdown automaton is working but an extra trick is required.

The trick is that we will only be interested in triples $(k_1,k_2,k_3)$ where $k_1,k_2$ and $k_3$ are even (and these can be easily filtered out by an MSO sentence, \emph{e.g.} considering the path from the initial configuration corresponding to $(0,0,0)$ and performing a modulo $2$ computation). We call these configurations \defin{even} configurations. With an even configuration and the corresponding triple  $(k_1,k_2,k_3)$ we can associate a pair of counters $(\frac{k_1-k_3}{2},\frac{k_2-k_3}{2})\in\mathbb{Z}\times\mathbb{Z}$.  We take the same pushdown automaton as previously but now in the first mode to simulate changes on the counter one must do two identical successive transitions: \emph{e.g.} to increment counter $k_1-k_3$ and decrement the counter $k_2-k_3$ one pushes $(3,1,2)$ twice.

Fix a well-formed marking $H$ of $G_{\mathcal{P}}$. 
Assume we are in an even configuration with an associated triple $(k_1,k_2,k_3)$ and that we want to check whether $k_1=k_3$. For this one looks at the marked outgoing edges in the present configurations: obviously, if both the one going to the state $p_1$ and the one going to the state $p_3$ are marked, one directly concludes that $k_1=k_3$, hence that $k_1-k_3=0$; however one may have $k_1=k_3$ without being in the previous situation namely if $k_2<k_1,k_3$ (as in the proof of Theorem~\ref{theo:full}) or if $k_2>k_1,k_3$, $k_1=k_3$ but only the edge to $p_1$ is marked (or symmetrically the one to $p_3$). We handle the case where $k_2<k_1,k_3$ as previously, \emph{i.e.} by looking at a path (of even length) along which only pushing $(2,3,2)$ are performed and that leads to a configuration with some property. The latter property is either that both the outgoing edge to $p_1$ and to $p_3$ are marked or that only the outgoing edge to $p_1$ is marked (\emph{resp.} $p_3$) and in the non-even configuration obtained by pushing {$(3,3,2)$} (\emph{resp.} $(1,3,2)$) the edge to $p_3$ (\emph{resp.} $p_1$) is marked: this means that $k_1\leq k_3$ (\emph{resp.} $k_3\leq k_1$) and that {$k_1+3\geq k_3+2$} (\emph{resp.} $k_3+2\geq k_1+1$) hence that $k_1=k_3$ as both are even. 

As the existence of a path (of even length) along which only pushing $(2,3,2)$ are performed and that leads to a configuration with the previous property, can easily be expressed as an MSO formula, and as being an even configuration can also be stated in MSO, it follows that one can write an MSO formula (with one order-$1$- free variable) that holds exactly in those even configurations (in the first mode) where $k_1=k_3$. Similarly, one can design a formula to check whether $k_2=k_3$. 

Then, as in the proof of Theorem~\ref{theo:full}, it follows that, for any 2-counter machine with tests for $0$, one can build an MSO formula that is true in $H$ if and only if the 2-counter machine halts from its initial configuration. As the halting problem is undecidable for 2-counter machine with tests for $0$, it follows that $H$ has an undecidable MSO theory.
\end{proof}

\section{Consequences}\label{ref:consequences}

We now briefly mention some implications of our results. 
As we know (by Theorem~\ref{theo:MSOpushdown}) that any pushdown graph has a decidable MSO theory, one directly gets the following.

\begin{corollary}
There is a pushdown automaton $\mathcal{P}$ such that no well-formed marking of it is a pushdown automaton.
\end{corollary}

A generalisation of reachability properties are given by the notion of pushdown reachability games. In the setting of pushdown graph (that for simplicity we assume without dead-end), one considers a partition of the control states of the underlying pushdown automaton {between} two players: Eve and Adam. A play consists in moving a pebble starting from an initial configuration as follows: at any stage of the game the player owning (the control state of) the current configuration chooses a successor and moves the pebble to it and so on forever. The play is won by Eve if the pebble eventually reaches a final configuration. A strategy for a player is a function mapping every prefix of plays to a valid move; a player respects a strategy along a play if he systematically plays the move indicated by the strategy; and a strategy is winning for a player from an initial vertex if the player wins any play starting from that vertex when he respects the strategy. 

It is a classical result that the winning region for Eve (\ie the set of configurations from which she has a winning strategy) can be defined in a fixpoint manner leading to an object called \defin{attractor} (see \eg \cite{Zie98}): this construction is a direct adaptation of the definition (in Section~\ref{ssec:pushdown}) of the sequence $(W_i)_{i\geq 0}$ to the alternating setting. In particular, if Eve owns all the control states, the attractor coincides with $W$. Of course, the attractor constructions also come with a corresponding (positional) strategy (namely a strategy that only depend on the current configuration and that ensures to get closer to a final configuration). Hence, one has the following immediate consequence of Theorem~\ref{theo:full} and Theorem~\ref{theo:sub}.

\begin{corollary}
There is a pushdown reachability game such that any marking of the underlying pushdown graph by an attractor strategy (or any sub-strategy of it) leads to a graph with an undecidable MSO theory.
\end{corollary}

{This result is interesting because it implies that the constructions in \cite{Walukiewicz01,SerrePHD} for pushdown games build a winning strategy that is not an attractor strategy. Indeed, the constructed strategies can be used to define a well-formed marking of the input pushdown graph which remains a pushdown graph.}

\section*{Acknowledgements}

The authors would like to thank Didier Caucal for numerous helpful discussions and for proof-reading an early version of this manuscript. {They would also like to thanks the reviewers for all their valuable comments.}

\section*{\refname}

\newcommand{\noopsort}[1]{} \newcommand{\singleletter}[1]{#1}
  \newcommand{\etal}{et al.}

\end{document}